\documentclass[aps,twocolumn,showpacs,superscriptaddress]{revtex4}
\usepackage{amsmath}
\usepackage{amsfonts}
\usepackage{amsthm}
\usepackage{amssymb}

\newcommand{\Span}{\ensuremath{\mathop{\text{Span}}\nolimits}}
\newcommand{\Hil}{\ensuremath{\mathcal{H}}}
\newcommand{\Id}{\ensuremath{\mathord{\mathit{Id}}}}
\renewcommand{\d}{\ensuremath{^\dagger}}
\newcommand{\C}{\ensuremath{\mathcal{C}}}
\newcommand{\tC}{\ensuremath{\mathcal{\tilde C}}}
\newcommand{\ttC}{\ensuremath{\mathcal{C'}}}
\newcommand{\T}{\ensuremath{\mathcal{T}}}

\newcommand{\AutGS}{\ensuremath{\mathord{\mathit{Aut}(G\mid S)}}}
\newcommand{\Rep}{\ensuremath{\mathord{\mathit{Rep}}}}

\newcommand{\ket}[1]{\ensuremath{\mathord{|#1\rangle}}}
\newcommand{\bra}[1]{\ensuremath{\mathord{\langle#1|}}}
\newcommand{\braket}[2]{\ensuremath{\mathord{\langle#1|#2\rangle}}}

\newtheorem{thm}{Theorem}
\newtheorem{lma}{Lemma}
\newtheorem{cor}{Corollary}

\theoremstyle{definition}
\newtheorem{dfn}{Definition}
\newtheorem*{rem}{Remark}
\newtheorem*{exm}{Example}

\newcommand{\secref}[1]{Section~\ref{sec-#1}}
\renewcommand{\eqref}[1]{Eq.~(\ref{eq-#1})}
\newcommand{\eqsref}[3]{Eqs.~(\ref{eq-#1}) #2~(\ref{eq-#3})}
\newcommand{\thmref}[1]{Theorem~\ref{thm-#1}}
\newcommand{\dfnref}[1]{Definition~\ref{dfn-#1}}
\newcommand{\lmaref}[1]{Lemma~\ref{lma-#1}}
\newcommand{\corref}[1]{Corollary~\ref{cor-#1}}

\newcommand{\beq}{\begin{equation}}
\newcommand{\eeq}{\end{equation}}
\renewenvironment{matrix}{\left(\!\begin{array}{ll}}{\end{array}\!\right)}

\divide\arraycolsep by 2

\begin{document}

\author{V. Poto\v cek}
\email{vaclav.potocek@fjfi.cvut.cz}
\affiliation{Department of Physics, FNSPE, Czech Technical University in        
Prague, B\v rehov\'a 7, 115 19 Praha, Czech Republic.}

\title{Symmetries in discrete time quantum walks on Cayley graphs}

\begin{abstract}
We address the question of symmetries of an important type of quantum walks.  
We introduce all the necessary definitions and provide a rigorous formulation of the 
problem. Using a thorough analysis, we reach the complete answer by presenting 
a constructive method of finding all solutions of the problem with minimal 
additional assumptions. We apply the results on an example of a quantum walk 
on a line to demonstrate the practical significance of the theory.
\end{abstract}


\maketitle

\section{Introduction}
\label{sec-intro}

The search for symmetries is an important problem in all fields of physics.  
In both classical and quantum mechanics, the knowledge of symmetries of 
a given system can help significantly in finding a solution of its equations 
of motion, in reducing the number of parameters, or identifying the integrals 
of motion.

In this paper, we aim to find the symmetries of the time evolution equation of 
a broad class of discrete time quantum walks. We note that this important 
question has been addressed partly by other authors. Symmetries of particular 
quantum walk scenarios have been classified, e.g., in \cite{chandra}. A special 
class of symmetries of discrete time quantum walks on Cayley graphs has been 
studied in \cite{sym-hit} in relation to global analytic properties of the 
quantum walks. Symmetries have played an essential role in an approximate 
analytic solution of time evolution in the Shenvi-Kempe-Whaley algorithm 
\cite{skw} for quantum database searching.  Another use of symmetries has been 
presented in a recent experimental realization of a quantum walk on a line 
\cite{exp} when a reduced set of parameters have been shown to cover all 
possible configurations of the model. However, no general study focused on the 
symmetries themselves has been presented so far.

The article is structured as follows.
In \secref{def}, we define the class of discrete time quantum walk to be 
studied in more detail.
In \secref{symm1}, we use a general method to find all symmetries of the time 
evolution equation which preserve measurement probabilities.
In \secref{symm2}, we extend the result by generalizing the notion of 
symmetries of the system to allow automorphisms of the underlying graph.
In \secref{concl}, we conclude and discuss our results.

\section{Quantum Walks on Cayley Graphs}
\label{sec-def}

In the scope of this paper, we will restrict our study to discrete time 
quantum walks on Cayley graphs, with the quantum coin reflecting the graph 
structure.  This class of graphs, however, covers all the most important cases 
used in algorithmic applications of quantum walks---lattices both with and 
without periodic boundary conditions \cite{lattice2}, hypercube graphs 
\cite{skw}, among many others.

In general, Cayley graphs are defined as follows:

\begin{dfn}
\label{dfn-cayley}
Let $G$ be a discrete group finitely generated by a set $S$. The (uncolored) 
\emph{Cayley graph} $\Gamma = \Gamma(G,S)$ is a directed graph $(G,E)$, where 
the set of vertices is identified with the set of elements of $G$ and the set 
of edges is
$$E = \lbrace (g,gs) \mid g \in G, s \in S\rbrace.$$
\end{dfn}

A discrete quantum walk on a given Cayley graph is defined as the time 
evolution of a particle confined to the vertices of the graph, and allowed to 
move along its edges, one per a discrete time step. Thus, the Hilbert space 
corresponding to the spatial degree of freedom of the particle is the space of 
$\ell^2$ functions defined on $G$, or equivalently, the space spanned by 
orthonormal basis states corresponding to the elements of $G$:
\begin{subequations}
\beq
\label{eq-hs}
\Hil_S = \ell^2(G) = \Span_{\mathbb{C}} \lbrace\,\ket{x} \mid x \in G\rbrace.
\eeq
Besides the spatial degree of freedom, we will require the particle undergoing 
the walk (the walker) to have an internal degree of freedom whose dimension 
equals the cardinality of $S$.
\beq
\label{eq-hc}
\Hil_C = \ell^2(S) = \Span_{\mathbb{C}} \lbrace\,\ket{c} \mid c \in S\rbrace.
\eeq
\end{subequations}
This is in a direct analogy to \cite{first} where quantum walks on general 
regular graphs have been introduced.  The need for the presence of an internal 
degree of freedom has been shown to be crucial for quantum walks on Euclidean 
lattices \cite{no-go} in order to reach nontrivial unitary time evolutions.  
A generalization of this ``No-Go Lemma'' to all Cayley graphs has been negated 
in \cite{cayley-scalar}. In the scope of this article, however, we will keep the 
assumption about the internal degree of freedom as stated above.

\begin{dfn}
\label{dfn-hil}
The Hilbert spaces defined in \eqsref{hs}{and}{hc} are called \emph{position} 
and \emph{coin Hilbert spaces}, respectively. The full state space of the 
system is then
$$\Hil = \Hil_S \otimes \Hil_C = \Span_{\mathbb{C}} \lbrace\,\ket{x,c} \mid 
x \in G, c \in S\rbrace,$$
where $\ket{x,c} = \ket{x}_S\otimes\ket{c}_C$. We will refer to the systems 
$\lbrace\,\ket{x} \mid x \in G\rbrace$, $\lbrace\,\ket{c} \mid c \in 
S\rbrace$, and $\lbrace\,\ket{x,c} \mid x \in G, c \in S\rbrace$, as to 
\emph{geometrical bases} of $\Hil_S$, $\Hil_C$ and $\Hil$, respectively.
\end{dfn}
In the following, the symbols $G$, $S$, $\Gamma$, $\Hil_S$, $\Hil_C$ and 
$\Hil$ will always denote the objects introduced in Definitions 
\ref{dfn-cayley} and~\ref{dfn-hil}. Moreover, a tensor product of two vectors 
or operators will be always understood to follow the factorization of $\Hil$ 
into $\Hil_S$ and $\Hil_C$.

This factorization of the state space plays a key role in the idea of 
a quantum walk.  The general assumption is that operations which keep the 
position of the walker intact are generally available, whereas the position 
register can only be affected via controlled transitions of the walker on the 
underlying graph $\Gamma$. We will formalize the former in the following 
definition:
\begin{dfn}
\label{dfn-local}
Let $A \in GL(\Hil)$. We will call $A$ a \emph{local operation} if and only if 
there is a map $\omega_A: G \to GL(\Hil_C)$ such that $A$ allows the following 
decomposition:
\beq
\label{eq-local}
A = \sum_{x \in G} \ket{x}\bra{x} \otimes \omega_A(x).
\eeq
\end{dfn}
It trivially follows that for each local operation $A$, the decomposition by
\eqref{local} is unique. Moreover, if $A \in U(\Hil)$, then all the components 
$\omega_A(x)$ are elements of $U(\Hil_C)$, and \emph{vice versa}. For a local 
operation $A$, we will use notation $A_x = \omega_A(x)$ for the components in 
this decomposition.

We note that the set of local operations depends not only on the separation of 
$\Hil$ to a tensor product of $\Hil_S$ and $\Hil_C$ but also on the choice of 
the basis in $\Hil_S$. In any case, however, the local operations form 
a subgroup of $GL(\Hil)$.

It is important to distinguish local operations on $\Hil$ from operations acting 
only on $\Hil_C$, that is, operators of the form $B = \Id \otimes B'$.  The 
latter form a subgroup of the group of local operations: indeed, any such $B$ is 
local with $B_x = B'$ for all $x \in G$.

Out of the other class of operations, altering the position of the walker, one 
representative is sufficient:
\begin{dfn}
\label{dfn-step}
The \emph{step operator} $T$ is a controlled shift operator on $\Hil_S$ 
conditioned by the coin register, as prescribed by its action on the basis 
states $\ket{x,c}$,
\beq
\label{eq-step}
T\ket{x,c} = \ket{xc,c}.
\eeq
\end{dfn}
Clearly, $T$ is defined by \eqref{step} on the whole of $\Hil$ via linearity 
and is a bounded operator. As the tensor product basis states, specified by 
\dfnref{hil}, are solely permuted under $T$, it is obvious that $T$ is an 
unitary operator on $\Hil$ and can thus form a time evolution operator in 
a discrete time quantum system.\footnote{%
We note that the above definition of a step operation on the Hilbert space 
$\Hil$ is not the only possible one; as shown in \cite{cayley}, the concept of 
quantum coin can be altered so that the basis coin states do not imply the 
transition over individual edges from the vertex $x$ in a one-to-one manner.  
Throughout the text, however, we will stay with \dfnref{step}.}

Before defining a quantum walk, we need one last supporting definition:
\begin{dfn}
\label{dfn-coin}
Let $\C = (C_n)_{n=0}^{+\infty}$ is an infinite sequence of local unitary 
operations on $\Hil$. We call $\C$ a \emph{quantum coin}. If the sequence is 
constant, we call the quantum coin $\C$ \emph{time-homogeneous}.  If every 
term $C_n$ of the sequence is a tensor product $\Id \otimes C'_n$, we call the 
quantum coin $\C$ \emph{space-homogeneous}. In general, however, a quantum 
coin may be both time- and position-dependent.
\end{dfn}
We call a generic $\C$ a time- and position-dependent coin since, in 
accordance with \dfnref{local}, we can find unitary operators $C_{n,x} \in 
U(\Hil_C)$ for each time $n \in \mathbb{N}_0$ and position $x \in G$ which 
alter the coin register in dependence on both the current time and the state 
of the position register, provided that the latter is well-defined.

The set of all quantum coins forms a group under element-wise composition.

The coin and step operators lead us to the definition of a discrete time 
quantum walk on a Cayley graph $\Gamma$.
\begin{dfn}
\label{dfn-qw}
Let $\Gamma$ is a Cayley graph, let $\C = (C_n)_{n=0}^{+\infty}$ is a quantum 
coin on its corresponding Hilbert space $\Hil$.  A \emph{discrete-time quantum 
walk} on $\Gamma$ with the coin $\C$ is a quantum protocol described by the 
following: an initial state $\ket{\psi_0} \in \Hil$ and the evolution operator
\begin{subequations}
\label{eq-qw-unrestr}
\beq
\label{eq-prop}
W_\C: \mathbb{N}_0 \to U(\Hil): W_\C(n) = TC_{n-1}TC_{n-2}\ldots TC_0.
\eeq
For $n \in \mathbb{N}_0$, we say that the state of the walker after $n$ steps is 
\beq
\label{eq-psin}
\ket{\psi_n} = W_\C(n) \ket{\psi_0}.
\eeq
\end{subequations}
\end{dfn}

\section{Symmetries Preserving Measurement Probabilities}
\label{sec-symm1}

Symmetry of a system is an invariance of the system under some kind of 
transformation acting on its parameters and/or the initial state.  Invariance 
does not necessarily mean that the time evolution is exactly the same, some 
variations may take place in the internal state as long as they do not 
influence the observable properties of the system, that is, the measurement 
probabilities of the spatial degree of freedom. Of all such transformations, 
we will be interested only in those which respect the unitary nature of 
quantum mechanics. Formally, we can state the requirement as follows:
\begin{dfn}
\label{dfn-symm}
Let $\T$ be an endomorphism on the Cartesian product of the set of quantum 
coins and initial states of a quantum walk on $\Gamma$. We call $\T$ 
a \emph{unitary quantum walk symmetry} on $\Gamma$ if there is a sequence of 
local unitary operators $(U_n)_{n=0}^{+\infty}$ such that for each quantum 
coin $\C = (C_n)_{n=0}^{+\infty}$ and for each initial state $\ket{\psi_0}$,
\beq
\label{eq-symm}
\forall n \in \mathbb{N}_0: \quad W_\tC(n) \ket{\tilde\psi_0} = U_n W_\C(n) 
\ket{\psi_0},
\eeq
where $\tC = (\tilde C_n)_{n=0}^{+\infty}$ and $\ket{\tilde\psi_0}$ denote the 
image of $\C$ and $\ket{\psi_0}$ under $\T$.
\end{dfn}

The above definition is motivated by the fact that local unitary operations 
preserve measurement probabilities in the geometrical basis of $\Hil_S$,
$$\sum_{c \in S} |\braket{x,c}{\phi}|^2 =: \|\braket{x}{\psi}\|^2 
= \|\bra{x}U_{\text{local}}\ket{\psi}\|^2.$$

\begin{lma}
\label{lma-diff}
Let, in the notation of \dfnref{symm}, $(\tC, \ket{\tilde\psi_0}) = \T(\C, 
\ket{\psi_0})$. Then the condition of \eqref{symm} is satisfied if and only if
\begin{subequations}
\beq
\label{eq-symm-diff-a}
\ket{\tilde\psi_0} = U_0 \ket{\psi_0},
\eeq
\vskip -\abovedisplayskip
\vskip -\belowdisplayskip
\beq
\label{eq-symm-diff-b}
\forall n \in \mathbb{N}_0: \quad T\tilde C_n = U_{n+1}TC_nU_n\d,
\eeq
\end{subequations}
\end{lma}

\begin{proof}
The first part is readily obtained by studying the special case of \eqref{symm} 
where $n=0$. Inserting \eqref{symm-diff-a} back into \eqref{symm}, we get
\beq
\label{eq-symmpom}
\forall n \in \mathbb{N}_0: \quad W_\tC(n) U_0 \ket{\psi_0} = U_n W_\C(n) 
\ket{\psi_0}
\eeq

The generality of \eqref{symmpom} with respect to $\ket{\psi_0}$ implies an 
equivalence of the operators,
$$W_\tC(n)U_0 = U_n W_\C(n).$$
Substituting $n+1$ for $n$, we get another identity,
$$W_\tC(n+1)U_0 = U_{n+1} W_\C(n+1).$$
Comparing with
$$\begin{array}{rcl}
W_\tC(n+1) &=& T\tilde C_nW_\tC(n), \\
W_\C(n+1) &=& TC_nW_\C(n),
\end{array}$$
we obtain the relation
$$T\tilde C_nU_nW_\C(n) = U_{n+1}TC_nW_\C(n).$$
Due to the unitarity of the time evolution operators and $U_n$, this is 
equivalent to \eqref{symm-diff-b}.
\end{proof}

\begin{lma}
\label{lma-diag}
Let $\T$ be a unitary quantum walk symmetry imposing a local unitary transform 
$(U_n)_{n=0}^{+\infty}$ on the instantaneous state of a quantum walk, as given 
by \dfnref{symm}. Then $U_{n,x}$ is diagonal in the geometrical basis of 
$\Hil_C$ for each $n \in \mathbb{N}$ (i.e. $n \ge 1$) and all $x \in G$, that 
is, there are complex units $u_{n,x,c}$ for each $n \in \mathbb{N}$, $x \in 
G$, and $c \in S$ such that
\beq
\label{eq-diag}
U_n = \sum_{x\in G} \sum_{c \in S} u_{n,x,c} \ket{x,c}\bra{x,c}
\eeq
\end{lma}

\begin{proof}
Starting from \eqref{symm-diff-b}, we can rearrange the terms so that $U_{n+1}$ 
is isolated:
$$U_{n+1} = T \tilde C_n U_n C_n\d T\d.$$
Let $x \in G$ and $c,d \in S$. We can compare the corresponding matrix 
elements on both sides:
$$\bra{x,c}U_{n+1}\ket{x,d} = \bra{x,c}T \tilde C_n U_n C_n\d T\d\ket{x,d}.$$
From \dfnref{step} and the subsequent comment, we can derive that 
\begin{subequations}
\beq
\label{eq-Td1}
T\d\ket{x,d} = \ket{xd^{-1},d}
\eeq
and similarly
\beq
\label{eq-Td2}
\bra{x,c}T = (T\d\ket{x,c})\d = \bra{xc^{-1},c}.
\eeq
\end{subequations}
Noting that all the other operators are local, we can factor out the position 
register to get
$$\begin{array}{l}
\braket{x}{x} \bra{c}U_{n+1,x}\ket{d} = \\
\qquad = \braket{xc^{-1}}{xd^{-1}} \bra{c} \tilde C_{n,xd^{-1}} U_{n,xd^{-1}} 
C\d_{n,xd^{-1}} \ket d
\end{array}$$
If $c \neq d$, the right hand side is zero due to its leftmost term.  Since 
$\braket{x}{x} = 1$, we obtain the implication
$$c \neq d\ \Rightarrow\ \bra{c}U_{n+1,x}\ket{d} = 0,$$
meaning that $U_{n+1,x}$ is diagonal in the geometrical basis of $\Hil_C$ for 
all $n \in \mathbb{N}_0$.

The second part of the Lemma is a trivial application of the corresponding 
definitions.
\end{proof}

\begin{thm}
\label{thm-unrestr}
Let $\T$ be a unitary symmetry of a quantum walk on $\Gamma$.
Then there is a unique local unitary operation $U_0$ on $\Hil$ and a unique 
sequence of local unitary operations $(U_n)_{n=1}^{+\infty}$ diagonal in the 
geometrical basis of $\Hil$ such that for each quantum coin $\C$ and each 
initial state $\ket{\psi_0}$, the transformed values read $\ket{\tilde\psi_0} 
= U_0\ket{\psi_0}$ and
\begin{subequations}
\label{eq-unrestr}
\beq
\forall n \in \mathbb{N}_0: \quad \tilde C_n = \sum_{x \in G} \left( 
\ket{x}\bra{x} \otimes \left( V_{n,x}C_{n,x}U_{n,x}\d \right) \right),
\eeq
where $V_{n,x} \in U(\Hil_C)$ is related to $U_{n+1}$ by
\beq
V_{n,x} = \sum_{c \in S} u_{n+1,xc,c} \ket{c}\bra{c},
\eeq
\end{subequations}
using the notation of \eqref{diag}. Conversely, given any $U_0$ and 
$(U_n)_{n=1}^{+\infty}$ satisfying the aforementioned conditions, there is 
a unique symmetry $\T$ yielding these values. Therefore, the symmetry group of 
\eqref{qw-unrestr} is $U(\Hil_C)^G \times U(1)^{G \times \mathbb{N} \times 
S}$.
\end{thm}

\begin{proof}
The proof follows from \eqref{symm-diff-b} and its equivalent form,
$$\tilde C_n = T\d U_{n+1} T C_n U_n\d.$$
Comparing the matrix elements, we obtain
$$\bra{x,c} \tilde C_n \ket{x,d} = \bra{x,c} T\d U_{n+1} T C_n U_n\d 
\ket{x,d}$$
Using \eqref{step} and the locality of the $U$ and $C$ operations, we find 
that
$$\bra{c}\tilde C_n \ket{d} = \bra{xc,c}U_{n+1}T\left(\ket{x} \otimes 
C_{n,x}U_{n,x}\d \ket{d}\right).$$
Using \lmaref{diag} and \eqref{Td1},
$$\begin{array}{l}
\bra{xc,c}U_{n+1}T = \left(T\d U_{n+1}\d \ket{xc,c}\right)\d = \\
\qquad = \left(u_{n+1,xc,c} T\d \ket{xc,c}\right)\d = u_{n+1,xc,c}^\ast 
\bra{x,c},
\end{array}$$
whence it follows that
$$\begin{array}{l}
\bra{c} \tilde C_n \ket{d} = u_{n+1,xc,c}^\ast \braket{x}{x} 
\bra{c}C_{n,x}U_{n,x}\d\ket{d} = \\
\qquad = \bra{c}V_{n,x}C_{n,x}U_{n,x}\d\ket{d},
\end{array}$$
as stated by the theorem.

Conversely, given the unitary operations $U_0$ and $(U_n)_{n=1}^{+\infty}$, 
\eqref{symm-diff-b} describes a way to construct a symmetry operation $\T$.
\end{proof}

According to \thmref{unrestr}, the sequence $(U_n)_{n=0}^{+\infty}$ provides 
a full classification of all the unitary quantum walk symmetries. If there is 
no restriction on the homogeneity of the quantum coins $\C$ and $\tC$ or the 
initial state, the choice of $U_{n}$ is free, up to the restriction of 
\lmaref{diag}.  More interesting cases arise when the coin has some global 
property that is required to be preserved under the symmetry.

Before stating the main theorem regarding homogeneous quantum coins, we 
introduce a means of classifying various walking spaces.
\begin{dfn}
\label{dfn-causal}
Let $G$ is a discrete group generated by a subset $S$, let $S^{-1}$ denote the 
set of inverses of all elements of $S$. The causal subgroup of $G$ with 
respect to $S$ is defined as
\begin{subequations}
\beq
\label{eq-causal}
S^{(0)} = \Big\langle \bigcup_{n \in \mathbb{Z}} S^nS^{-n} \Big\rangle.
\eeq
The future causal subgroup of $G$ with respect to $S$ is defined as
\beq
\label{eq-causal2}
S^{(0)}_+ = \Big\langle \bigcup_{n=1}^{+\infty} S^nS^{-n} \Big\rangle.
\eeq
\end{subequations}
A Cayley graph $\Gamma = \Gamma(G,S)$ is called nonseparating if $S^{(0)}_+$ 
= $S^{(0)}$.
\end{dfn}

In other words, the causal subgroup $S^{(0)}$ contains all the elements of $G$ 
which can be written as a product of generators and their inverses in such 
a way that the exponents add up to zero. The causal subgroup has several 
important properties, as shown in the following Theorem.
\begin{thm}
\label{thm-causal}
The causal group $S^{(0)}$ is a normal subgroup of $G$. Moreover, $G/S^{(0)}$ 
is a cyclic group generated by the coset of any element in $S$.
\end{thm}

\begin{proof}
Let $c \in S$, let $s \in S^nS^{-n}$ for some $n \in \mathbb{Z}$. Then it is 
simple to show that both $csc^{-1}$ and $c^{-1}sc$ are elements of $S^{(0)}$.  
Indeed, let $n > 0$.  Then $csc^{-1} \in SS^nS^{-n}S^{-1} = S^{n+1}S^{-(n+1)} 
\subset S^{(0)}$.  Similarly, $c^{-1}sc \in S^{-1}S^nS^{-n}S 
= (S^{-1}S^1)(S^{n-1}S^{-(n-1)})(S^1S^{-1}) \subset S^{(0)}$. The case $n < 0$ 
is analogous, $n = 0$ is trivial.

Using elementary algebra, this result can be generalized to any $c \in G$ and 
$s \in S^{(0)}$, which is one of the conditions equivalent to $S^{(0)}$ being 
normal in $G$.

For the second part, let $c_0$ be an arbitrary fixed element of $S$. We first 
show that the coset $cS$ equals $c_0S$ for any $c \in S$. Indeed,
$$cS = (c_0c_0^{-1})cS = c_0\underbrace{(c_0^{-1}c)}_{\in S}S = c_0S.$$
Analogously, $c^{-1}S = c_0^{-1}S$.

Let now $g$ be an arbitrary element of $G$. We can decompose $g$ into
$$g = c_1^{\epsilon_1}c_2^{\epsilon_2}\ldots c_k^{\epsilon_k},$$
where $c_i \in S$ and $\epsilon_i \in \mathbb{Z}$ for all $1 \le i \le k$.  
Using the above result, the coset $gS$ is equal to
$$gS = c_0^{\epsilon_1}c_0^{\epsilon_2}\ldots c_0^{\epsilon_k}S 
= c_0^{\epsilon_1+\epsilon_2+\ldots+\epsilon_k}S 
= (c_0S)^{\epsilon_1+\epsilon_2+\ldots+\epsilon_k}.$$
This completes the proof.
\end{proof}

\begin{rem}
The future causal subgroup $S^{(0)}_+$ generally does not share these 
properties. As they are extremely helpful for the theorems to follow, we will 
restrict the analysis below to quantum walks on nonseparating Cayley graphs, 
where there is no difference between $S^{(0)}_+$ and $S^{(0)}$.
\end{rem}

We note without proof that a sufficient condition for the equality $S^{(0)}_+ 
= S^{(0)}$ is that for each $c,d \in S$, $c^{-1}d$ is an element of 
$S^{(0)}_+$. This is satisfied automatically in, but not restricted to, all 
abelian groups. On the other hand, an example that this property is not 
universal is provided by the free group on $2$ or more generators. In such 
cases, the quantum walk splits the initial excitation into a potentially 
unlimited number of mutually independent branches which never can interfere 
again.

In the following, we denote $[G:S^{(0)}] = \chi(G,S)$. This characteristic 
plays its role in an important corollary of \thmref{causal}:

\begin{cor}
\label{cor-decomp}
Let $c_0$ be a fixed element of $S$. For each $x \in G$, there is $\tilde 
x \in S^{(0)}$ and $k \in \mathbb{Z}$ such that $x = \tilde x c_0^k$. This 
decomposition is unique if and only if $[G:S^{(0)}]$ is infinite, otherwise 
$k$ is determined up to an integer multiple of $\chi(G,S)$.
\end{cor}

Let $\T$ be a unitary quantum walk symmetry, as defined in \dfnref{symm}. From 
\thmref{indep}, we know that the quantum coin and the initial state are 
transformed independently. The following theorem studies two important cases 
where the transformation of the coin is restricted.

Let $\C$ denote a quantum coin and $\tC$ its image under $\T$. We say that 
$\T$ preserves time or space homogeneity of the quantum coin if the respective 
property of $\C$ implies that the same property is held for $\tC$.

\begin{thm}
\label{thm-indep}
\begin{subequations}
Let $\T$ be a unitary symmetry of a quantum walk on a nonseparating Cayley 
graph, let $(U_n)_{n=0}^{+\infty}$ be the transformation induced in the 
instantaneous state of the quantum walk.
\begin{itemize}
\item $\T$ preserves space homogeneity of the quantum coin if and only if the 
unitary operators $U_{n,x}$, forming the decomposition of $U_n$, are of the 
form
\beq
U_{n,\tilde x c_0^k} = \eta_{n-k} \rho(\tilde x) U'_{n},\ \forall n \in 
\mathbb{N}_0,
\eeq
where $(\eta_m)_{m \in \mathbb{Z}}$ is an arbitrary doubly infinite sequence 
of complex units, periodic with the period $\chi(G,S)$ if the latter is 
finite, $\rho(s)$ is a one-dimensional unitary representation of $S^{(0)}$ and 
the operators $U'_n$ act on $\Hil_C$ only.  The group of symmetries preserving 
space homogeneity is $(U(1)^{\chi(G,S)}/U(1)) \times \Rep(S^{(0)}) \times 
U(1)^{\mathbb{N} \times S} \times U(\Hil_C)$, where $\Rep(S^{(0)})$ is the 
group of one-dimensional unitary representations of $S^{(0)}$ with pointwise 
multiplication.
\item $\T$ preserves time homogeneity of the quantum coin if and only if the 
unitary operations $U_{n,x}$ are restricted by
\beq
U_{n,x} = \eta_{n-k} \epsilon^n U_x,\ \forall n \in \mathbb{N}_0,
\eeq
where $\eta_m$ is defined the same way as above, $\epsilon$ is an arbitrary 
complex unit and $U_x$ are the components of a unitary operation $U \in 
U(\Hil)$ diagonal in the geometrical basis of $\Hil$.  If $\chi(G,S)$ is 
infinite, we can take $\epsilon$ fixed at $1$.  The group of symmetries 
preserving time homogeneity is $(U(1)^{\chi(G,S)}/U(1)) \times U(1) \times 
U(1)^{G \times S}$ if $\chi(G,S) < +\infty$ and $(U(1)^{\chi(G,S)}/U(1)) \times 
U(1)^{G \times S}$ otherwise.
\end{itemize}
\end{subequations}
\end{thm}

\begin{proof}
In both cases, we start from \eqref{unrestr}. Let $x \in G$ and $c,d \in S$.  
Comparing matrix elements on both sides, we obtain
$$\bra{x,c}\tilde C_n\ket{x,d} = \bra{c}\tilde C_{n,x}\ket{d} 
= \bra{c}V_{n,x}C_{n,x}U_{n,x}\d\ket{d}.$$
Using the result of \lmaref{diag}, the right hand side can be simplified 
substantially:
\beq
\label{eq-ucka}
\bra{c}\tilde C_{n,x}\ket{d} = \frac{u_{n,xc,c}}{u_{n,x,d}}
\bra{c}C_{n,x}\ket{d}.
\eeq
Let $e$ be another element of $S$. We compare the last equation with another 
matrix element equation,
$$\bra{c}\tilde C_{n,x}\ket{e} = \frac{u_{n,xc,c}}{u_{n,x,e}}
\bra{c}C_{n,x}\ket{e}.$$
As these formulas hold for any quantum coin $\C$, we select one for which all 
the matrix elements of $C_{n,x}$ are nonzero for all $n \in \mathbb{N}$ and $x 
\in G$. We can then divide the above two equations to get
\beq
\label{eq-udivu}
\frac{u_{n,x,e}}{u_{n,x,d}} = \frac{\bra{e}C_{n,x}\ket{c}}{\bra{e}\tilde 
C_{n,x}\ket{c}} \frac{\bra{d}\tilde C_{n,x}\ket{c}}{\bra{d}C_{n,x}\ket{c}}.
\eeq

\textbf{Part A.}
In the case of a space-homogeneous coin, the right hand side of \eqref{udivu} 
is constant in $x$. This allows us to factorize $u_{n,x,c}$ into
$$u_{n,x,c} = v_{n,x}\delta_{n,c}.$$
Without loss of generality, we will require both the factors to have a modulus 
of $1$. Similarly, we find from \eqref{ucka} that the ratio 
$v_{n+1,xc}/v_{n,x}$ does not depend on $x$, that is, for each $y \in G$,
\beq
\label{eq-xhom-n+1}
\frac{v_{n+1,yc}}{v_{n,y}} = \frac{v_{n+1,xc}}{v_{n,x}}, \quad 
\frac{v_{n+1,yc}}{v_{n+1,xc}} = \frac{v_{n,y}}{v_{n,x}}.
\eeq
Let $k \in \mathbb{N}$, let $c_1,\ldots,c_k,d_1,\ldots,d_k \in S$.  By 
repeated use of \eqref{xhom-n+1}, we find that
$$\begin{array}{l}
\frac{v_{n,y}}{v_{n,x}} = \frac{v_{n+1,yc_1}}{v_{n+1,xc_1}} = \ldots 
= \frac{v_{n+k,yc_1c_2\ldots c_k}}{v_{n+k,xc_1c_2\ldots c_k}} = \\
\quad = \frac{v_{n+k-1,yc_1c_2\ldots c_kd_1^{-1}}}{v_{n+k-1,xc_1c_2\ldots 
c_kd_1^{-1}}} = \ldots = \\
\quad = \frac{v_{n,yc_1c_2\ldots c_kd_1^{-1}\ldots 
d_k^{-1}}}{v_{n,xc_1c_2\ldots c_kd_1^{-1}\ldots d_k^{-1}}}
\end{array}$$
Hence for all $n \in \mathbb{N}$, $x,y \in G$ and $s \in S^{(0)}_+ = S^{(0)}$,
\beq
\label{eq-xhom-rep}
\frac{v_{n,ys}}{v_{n,y}} = \frac{v_{n,xs}}{v_{n,x}}.
\eeq

Let $s,s' \in S^{(0)}$. By putting $x = e$---the identity element of $G$---and 
$y = s'$, we obtain from \eqref{xhom-rep}
$$\frac{v_{n,s's}}{v_{n,s'}} 
= \frac{\frac{v_{n,s's}}{v_{n,e}}}{\frac{v_{n,s'}}{v_{n,e}}} 
= \frac{v_{n,s}}{v_{n,e}}.$$
The last equation implies that $v_{n,s}/v_{n,e}$ for fixed $n \in \mathbb{N}$ 
is a homomorphism from $S^{(0)}$ to $U(1)$ and thus a one-dimensional unitary 
representation of $S^{(0)}$. Let us call this representation $\rho_n$.

Let now $x$ be a general element of $G$. By \corref{decomp}, taking any fixed 
$c_0 \in S$, we can find $k \in \mathbb{Z}$ and $\tilde x \in S^{(0)}$ such 
that $x = \tilde x c_0^k$. We then find $v_{n,x}$ to equal
$$v_{n,x} = v_{n,\tilde x c_0^k} = \alpha_{n,k} \rho_n(\tilde x).$$
Inserting this form into \eqref{xhom-n+1}, we find that the expression
$$\frac{v_{n+1,xc}}{v_{n,x}} = \frac{v_{n+1,\tilde x c_0^k c}}{v_{n,\tilde 
x c_0^k}} = \frac{\alpha_{n+1,k+1} \rho_{n+1}(\tilde x \overbrace{c_0^k 
c c_0^{-(k+1)}}^{\in S^{(0)}})}{\alpha_{n,k} \rho_n(\tilde x)}$$
is independent of $x$, that is, of both $\tilde x$ and $k$. This implies that 
$\rho_n$ is constant in $n$, so that we can call it $\rho$. Moreover,
$$\frac{\alpha_{n+1,k+1}\rho(c_0^k c c_0^{-(k+1)})}{\alpha_{n,k}}$$
must be constant in $k$.

By choosing $c = c_0$, the last equation becomes 
${\alpha_{n+1,k+1}}/{\alpha_{n,k}} = \beta_n$, whence we obtain
$$\alpha_{n,k} = \beta_0\beta_1\ldots \beta_{n-1}\alpha_{0,n-k} = \gamma_n 
\tilde\alpha_{n-k}.$$

If $\chi(G,S)$ is a finite number, the decomposition of \corref{decomp} is not 
unique. The value of $\chi(G,S)$ is then equal to the least positive power $l$ 
for which $c_0^l \in S^{(0)}$. Let $\tilde x_0 = c_0^{\chi(G,S)}$.  The equality
$$\tilde x c_0^k = \tilde x \tilde x_0 c_0^{k-\chi(G,S)}$$
then imposes a condition on the choice of $\alpha_{n,k}$ and subsequently 
$\tilde\alpha_m$:
$$\begin{array}{l}
\alpha_{n,k} \rho(\tilde x) = \alpha_{n,k-\chi(G,S)} \rho(\tilde x) 
\rho(\tilde x_0) \\
\quad \Rightarrow \alpha_{n,k+\chi(G,S)} = \alpha_{n,k} \rho(\tilde x_0) \\
\quad \Rightarrow \tilde\alpha_{m-\chi(G,S)} = \tilde\alpha_m \rho(\tilde 
x_0).
\end{array}$$
In this case, the freedom in choosing $\tilde\alpha$ is restricted to 
$\chi(G,S)$ independent complex units. If $\chi(G,S)$ is infinite, all 
elements of the doubly infinite sequence can be chosen freely.

Putting together all the above elements, we find that the complete solution of 
\eqref{udivu} with the right hand side independent of $x$ can be written as
\beq
\label{eq-xhom-sol}
u_{n,x,c} = \gamma_n \tilde\alpha_{n-k} \rho(\tilde x) \delta_{n,c}
\eeq
for $x = \tilde x c_0^k$, where
\begin{itemize}
\item $\gamma_n$ and $\delta_{n,c}$ are any complex units for all $n \in 
\mathbb{N}$ and $c \in S$,
\item $\tilde\alpha_m$ is a sequence of $\chi(G,S)$ independent complex units,
\item $\rho$ is a one-dimensional unitary representation of $S^{(0)}$.
\end{itemize}
Clearly, the sequence $\gamma_n$ can be absorbed into $\delta_{n,c}$.  Besides 
that, only one degree of freedom is counted twice---a global phase factor, 
which can come from both $\tilde\alpha$ and $\delta$.

At this point, we emphasize that the parameter $n$ so far has been greater 
than or equal to $1$; \lmaref{diag} puts no restriction on the form of $U_0$ 
except that it is local.  Thus the case $n=0$ must be studied separately.  
According to \lmaref{diff}, the transformation of the quantum coin element 
$C_0$ reads
$$\tilde C_0 = T\d U_1 T C_0 U_0\d.$$
Expressing $U_0$, we obtain
\beq
\label{eq-u0}
U_0 = \tilde C_0\d T\d U_1 T C_0.
\eeq
Comparing the corresponding matrix elements on both sides and expanding the 
matrix product on the right hand side while using the locality property of the 
$C$ and $U$ matrices gives
$$\bra{c}U_{0,x}\ket{d} = \sum_{a \in S} \bra{c} \tilde C_0\d \ket{a} 
u_{1,xa,a} \bra{a} C_0 \ket{d}.$$
This relates the components of $U_0$ to those of $U_1$, which are described by 
\eqref{xhom-sol}. Inserting the final form, we can see that
$$\begin{array}{l}
\bra{c}U_{0,x}\ket{d} = \sum_{a \in S} \bra{c} \tilde C_0\d \ket{a} 
\tilde\alpha_{1-k} \rho(\tilde x) \delta_{1,a} \bra{a} C_0 \ket{d} = \\
\quad = \alpha_{1-k} \rho(\tilde x) f(c,d),
\end{array}$$
where $x = \tilde x c_0^k$ and $f$ represents the matrix elements of some 
unitary matrix (any matrix can be reached with a suitable choice of $\tilde 
C_0$).  Therefore, the components of $U_0$ are complex unit multiples of one 
constant unitary operator on $\Hil_C$, where the dependence on $x$ follows the 
same rule as in the case of any other $U_n$, $n \ge 1$.

We conclude that the symmetry group under the aforementioned conditions is 
isomorphic to
$$(U(1)^{\chi(G,S)} / U(1)) \times Rep(S^{(0)}) \times U(1)^{\mathbb{N} \times 
S} \times U(\Hil_C),$$
as stated by the theorem.

\goodbreak
\textbf{Part B.}
If $\C$ and $\tC$ are simultaneously time-homogeneous, the right hand side of 
\eqref{udivu} is constant in $n$, which leads to a factorization $$u_{n,x,c} 
= w_{n,x} \delta_{x,c},$$
where we again assume both terms to be complex units. \eqref{ucka} then gives 
for $w_{n,x}$ that the ratio $w_{n+1,xc}/w_{n,x}$ does not depend on $n$ and 
thus for each $x \in G$ and $m,n \in \mathbb{N}_0$,
\begin{subequations}
\label{eq-t-indep-ucka}
\beq
\label{eq-thom-n+1}
\frac{w_{m+1,xc}}{w_{m,x}} = \frac{w_{n+1,xc}}{w_{n,x}}, \quad
\frac{w_{m+1,xc}}{w_{n+1,xc}} = \frac{w_{m,x}}{w_{n,x}}.
\eeq
In a complete analogy to the above, we obtain for each $x \in G$, $n \in 
\mathbb{N}_0$, and $s \in S_+$,
\beq
\frac{w_{m,xs}}{w_{m,x}} = \frac{w_{n,xs}}{w_{n,x}}.
\eeq
\end{subequations}
This means that for each $m$ and $n$ in $\mathbb{N}_0$ and each right coset 
$xS^{(0)}_+$, the ratio between $w_{m,y}$ and $w_{n,y}$ is a constant complex 
unit for all $y \in xS^{(0)}_+$, so that we can factorize
\beq
\label{eq-w-factor}
w_{n,x} = \alpha(n,xS^{(0)}_+) q_x.
\eeq
Once more, we will require both factors to be unitary. If, by assumption, 
$S^{(0)}_+ = S^{(0)}$, the cosets are identified by the power of one generator 
of~$S$---$c_0$---as $c_0^k S^{(0)}$, where $k \in \mathbb{Z}$, so that we obtain
$$w_{n,x} = w_{n,\tilde x c_0^k} = \alpha_{n,k} q_x.$$
We note that if $\chi(G,S)$ is finite, then $\alpha_{n,k+\chi(G,S)}$ must be 
equal to $\alpha_{n,k}$ to retain consistency. Inserting this form into 
\eqref{thom-n+1}, we find that the ratio $$\frac{\alpha_{n+1,k+1} 
q_{xc}}{\alpha_{n,k} q_x}$$
should not depend on $n$. This is equivalent to the condition that 
$\alpha_{n+1,k+1}/\alpha_{n,k}$ depends on $k$ only. Denoting this ratio 
$\beta_k$,
we find that
$$\begin{array}{l}
\alpha_{n,k} = \beta_{k-1} \alpha_{n-1,k-1} = \beta_{k-1} \beta_{k-2} 
\alpha_{n-2,k-2} = \\
\quad = \ldots = \beta_{k-1}\beta_{k-2}\ldots \beta_{k-n} \alpha_{0,k-n}
\end{array}$$
Denoting
$$\gamma_K = \begin{cases}
\prod_{k=0}^{K-1} \beta_k & \mbox{for}\ K \ge 0, \\
\prod_{k=1}^{-K} \beta_{-k}^{-1} & \mbox{otherwise},
\end{cases}$$
we can write
$$\alpha_{n,k} = \frac{\gamma_k}{\gamma_{k-n}}\alpha_{0,k-n}.$$
Unlike $\alpha_{0,k}$, $\gamma_k$ is not constant on the modular class $\bmod 
\chi(G,S)$ for $\chi(G,S) < +\infty$. Instead,
$$\gamma_{m+\chi(G,S)} = \prod_{k=0}^{\chi(G,S)-1} \beta_k \gamma_m =: 
\epsilon^{\chi(G,S)} \gamma_m.$$
In the case of infinite $\chi$, let $\epsilon = 1$. This allows us to write 
the solution uniformly as
$$\alpha_{n,k} = \epsilon^n \eta_{n-k} \gamma_k,$$
\beq
\label{eq-thom-sol}
u_{n,x,c} = \epsilon^n \eta_{n-k} \gamma_k \delta_{x,c},
\eeq
where $x = \tilde x c_0^k$ and
\begin{itemize}
\item $\delta_{x,c}$ are arbitrary complex units for all $x \in G$, $c \in S$,
\item $\gamma_m$ and $\eta_m$ are arbitrary sequences of $\chi(G,S)$ complex 
units,
\item $\epsilon$ is an arbitrary complex unit in the case of finite 
$\chi(G,S)$ and $1$ otherwise.
\end{itemize}
Again, as the term of $\gamma_k$ depends only on $x$, it can be immersed into 
$\delta_{x,c}$. Also, a global phase factor can be factored out of $\eta_m$ 
and put into $\delta_{x,c}$.

As opposed to the previous case, it's simple to determine the zeroth element 
$U_0$: starting from \eqref{u0}, we note that for time-homogeneous coins, 
there is a local unitary $C$ such that $C_n = C$ for all $n \in \mathbb{N}_0$.  
Similarly, $\tilde C_n = \tilde C$ for all $n \in \mathbb{N}_0$. Thus,
$$U_0 = \tilde C\d T\d U_1 T C.$$
We can compare this equation with its equivalent for $n=1$,
$$U_1 = \tilde C\d T\d U_2 T C.$$
Noting that by \eqref{thom-sol}, $U_{2,\tilde x c_0^k} = \epsilon 
\frac{\eta_{2-k}}{\eta_{1-k}} U_{1,\tilde x c_0^{k-1}}$, we obtain
$$U_{0,\tilde x c_0^k} = \epsilon^{-1} U_{1,\tilde x c_0^{k+1}},$$
so that the operator $U_0$ is also diagonal in the geometrical basis of $\Hil$ 
and its matrix elements are given simply by extending the validity of 
\eqref{thom-sol} to the case $n=0$.

We conclude the proof by establishing the group of time-homogeneity preserving 
symmetries. Taking into account \eqref{thom-sol} and the following notes, each 
symmetry is determined by specifying $\delta_{x,c}$, $\eta_m$ (up to a global 
phase) and possibly $\epsilon$. As all of these parameters are just tuples of 
complex units, this immediately gives the group in the form stated by the 
theorem.
\end{proof}

In the case of a both space- and time-homogeneous coin, we can easily combine 
the partial results given by \thmref{indep} as follows.
\begin{cor}
\label{cor-indep}
Let $\Gamma$ be nonseparating, let $\T$ be a unitary quantum walk symmetry 
described by a sequence of unitary operators $(U_n)_{n=0}^{+\infty}$. Then 
$\T$ preserves both time and space homogeneity of the quantum coin if and only 
if the components of $U_n$ are of the form
\beq
U_{n,x} = \eta_{n-k} \epsilon^n \gamma(x) U'
\eeq
for all $n \in \mathbb{N}_0$, where $\eta_m$ is defined the same way as in 
\thmref{indep}, $\epsilon$ is a complex unit, fixed at $1$ in the case where 
$\chi(\Gamma)$ is infinite, $\gamma$ is a one-dimensional unitary 
representation of $G$, and $U' \in U(\Hil_C)$ is a unitary operation diagonal 
in the geometrical basis of $\Hil_C$. The group of symmetries with this 
restriction is
$$(U(1)^{\chi(G,S)}/U(1)) \times U(1) \times \Rep(S^{(0)}) \times U(1)^S$$
if $\chi(\Gamma) < +\infty$ and
$$(U(1)^{\chi(G,S)}/U(1)) \times \Rep(S^{(0)}) \times U(1)^S$$
otherwise.
\end{cor}

\begin{exm}
In this example, we apply the above theory to a quantum walk on a line, where 
$G$ is the additive group of integers, $\mathbb{Z}$, generated by $S = \lbrace 
-1,1 \rbrace$, with a homogeneous coin. Even in this simplest case the above 
theory produces useful results. Let $\Gamma$ denote the Cayley graph 
$\Gamma(\mathbb{Z},S)$.

A general quantum coin with this property is given by $\C = (\Id \otimes 
C)_{n=0}^{+\infty}$, where $C$, expressed in the geometrical basis of $\Hil_C$, 
is a general unitary matrix of rank $2$,
\beq
\label{eq-gen-uni}
C = \omega\begin{matrix}
\mu & 0 \\ 0 & \mu^\ast
\end{matrix} \begin{matrix}
\phantom{-}\cos\psi & \sin\psi \\
-\sin\psi & \cos\psi
\end{matrix} \begin{matrix}
\nu & 0 \\ 0 & \nu^\ast
\end{matrix}.
\eeq
Here $\omega, \mu, \nu \in \mathbb{C}$, $\psi \in \mathbb{R}$, $|\omega| 
= |\mu| = |\nu| = 1$.

The causal subgroup is equal to $2\mathbb{Z}$, because any product of an odd 
number of generators is an even number, and $2$ can be written as $c+(-d) \in 
S+(-S) \subset S^{(0)}$ if $c=1$, $d=-1$. The condition of $\Gamma$ being 
nonseparating is a trivial property of any abelian walking space. We note that 
$\chi(G,S) = 2$ and the elements of $G:S^{(0)}$ correspond to the subsets of 
even and odd integers. Indeed, walks started in either of these subsets never 
interfere.

A general form of a unitary representation of $2\mathbb{Z}$ on $\mathbb{C}$ is
$$\gamma(x) = e^{i \phi x},\ \phi \in \mathbb{R}.$$
According to \corref{indep}, the symmetries of the above system are classified 
by five continuous parameters: $\eta_{\mathrm{odd}}/\eta_{\mathrm{even}}, 
\epsilon, \phi, \delta_{+1}, \delta_{-1}$. The transformed coin reads
\begin{subequations}
\beq
\label{eq-ex-trans-coin}
\begin{array}{l}
\tilde C = \omega\epsilon
\begin{matrix}
e^{i\phi} & 0 \\ 0 & e^{-i\phi}
\end{matrix} \begin{matrix}
\delta_{+1} & 0 \\ 0 & \delta_{-1}
\end{matrix} \begin{matrix}
\mu & 0 \\ 0 & \mu^\ast
\end{matrix} \cdot \\
\qquad \cdot \begin{matrix}
\phantom{-}\cos\psi & \sin\psi \\
-\sin\psi & \cos\psi
\end{matrix} \begin{matrix}
\nu & 0 \\ 0 & \nu^\ast
\end{matrix} \begin{matrix}
\delta_{+1}^\ast & 0 \\ 0 & \delta_{-1}^\ast
\end{matrix}\end{array}
\eeq
and the transformed initial state is
\beq
\label{eq-ex-trans-psi0}
U_0\ket{\psi_0} = \sum_{x \in \mathbb{Z}} \eta_{(x \bmod 2)} e^{i\phi x}
\begin{matrix} \delta_{+1} & 0 \\ 0 & \delta_{-1} \end{matrix}
\ket{x} \braket{x}{\psi_0}.
\eeq
\end{subequations}
Based on these formulas, some of the parameters assume a straightforward 
mathematical meaning:
\begin{itemize}
\item $\epsilon$ is related to the invariance of the system with respect to 
multiplying $C$ by a scalar. This is a phase that the system accumulates per 
every step of the quantum walk.
\item A common prefactor of $\delta_\pm$ is related to the freedom of global 
phase of the initial state.
\end{itemize}
The global phase can be completely moved from $\delta_{\pm}$ into $\eta_\xi$ 
by introducing a constraint $\delta_{-1} = \delta_{+1}^\ast$ and making 
$\eta_{\mathrm{even}}$ and $\eta_{\mathrm{odd}}$ two independent parameters.

In general, any continuous symmetry can be used to reduce the number of 
parameters determining nonequivalent instances of a given physical system. In 
our example, by choosing appropriate values of $\epsilon$, $\phi$, and 
$\delta_{\pm1}$, we can find a quantum walk equivalent with $W_\C$ in which 
the coin is simplified to
\beq
\label{eq-real}
\tilde C = \begin{matrix}
\phantom{-}\cos\psi & \sin\psi \\
-\sin\psi & \cos\psi
\end{matrix}
\eeq
and thus determined by a single parameter. The rest of the information about 
the particular quantum walk can be encoded into the initial state.

Besides this result, \eqref{real} has one nontrivial consequence: the 
transformed coin is a real-valued matrix and so is the infinite matrix of the 
step operator in the geometrical basis of $\Hil$. Therefore, an initial state 
with real coefficients in the geometrical basis will stay real-valued during 
the whole time evolution and an analogical result holds for a pure 
imaginary-valued initial vector. As a consequence, the real and imaginary 
parts of the transformed initial state define two quantum walks which never 
interfere, although visiting the same set of vertices. The contributions to 
measurement probabilities can be computed separately in the field of real 
numbers and classically summed.

Moreover, if the initial state of the walker is localized at a vertex $x_0$, 
i.e., of the form
$$\ket{\psi_0} = \ket{x_0} \otimes \ket{\chi_0},$$
then this property is kept under the transformation \eqref{ex-trans-psi0}. If 
we also neglect the global phase, which can be done using $\eta_\xi$ with no 
effect on the coin, the initial state is completely determined by two 
parameters (the spherical angles on the Bloch sphere). Thus any quantum walk 
on a line with position- and time-independent coin starting from a state 
localized at a given position is completely determined by a total of three 
degrees of freedom. This particular result has been exploited in a recent 
experimental realization \cite{exp} where there was only one adjustable 
parameter of the quantum coin, corresponding precisely to $\psi$ in this 
example, and a full control of the initial chirality $\chi_0$ (up to a global 
phase) using two adjustable optical elements.
\end{exm}

\section{Symmetries Involving Permutation of the Measurement Probabilities}
\label{sec-symm2}

In order to extend the applicability of the theory, we generalize the notion 
of quantum walk symmetries. According to \dfnref{symm}, the probability 
distribution of a complete measurement of the position register was required 
to stay invariant under a symmetry transformation. We obtain a broader class 
of solutions if we allow transformations which do affect the probability 
distribution, but in such a way that the original distribution is easily 
reconstructible---more precisely, such that the probabilities merely undergo 
some fixed permutation.  In order to respect the underlying group structure of 
the Cayley graph, we assume that the permutation is given by an automorphism
on $G$ and an optional multiplication by a fixed element of $G$, and define 
a wider class of symmetries which impose this kind of transformation on the 
measurement probability.

\begin{dfn}
\label{dfn-perm}
\begin{subequations}
Let $\phi$ be an automorphism of $G$ such that $\phi(S) = S$, let $g \in G$.  We 
call the map $g\phi: G \to G: x \mapsto g \cdot \phi(x)$ a \emph{shifted 
S-preserving automorphism} on $G$. We associate three operators with $g\phi$:
a \emph{spatial permutation operator} $P_{g\phi}^{(S)}$ on $\Hil_S$, defined 
by its action on geometrical basis states
\beq
P_{g\phi}^{(S)} \ket{x} = \ket{g\phi(x)}
\eeq
for all $x \in G$;
a \emph{coin permutation operator} $P_{g\phi}^{(C)}$ on $\Hil_C$, defined by
\beq
P_{g\phi}^{(C)} \ket{c} = \ket{\phi(c)}
\eeq
for all $c \in S$; and
a \emph{total permutation operator}
\beq
\label{eq-perm-total}
P_{g\phi} = P_{g\phi}^{(S)} \otimes P_{g\phi}^{(C)}
\eeq
on $\Hil$.
\end{subequations}
\end{dfn}
Note that the automorphism part $\phi$ of a shifted $S$-preserving 
automorphism $g\phi$, needed in the definition of $P_{g\phi}^{(C)}$, can be 
extracted using
$$\phi(c) = (g\phi(e))^{-1} (g\phi(c)).$$

\begin{dfn}
\label{dfn-symm2}
Let $\T$ be an endomorphism on the Cartesian product of the set of quantum 
coins and initial states of a quantum walk on $\Gamma$. We call $\T$ 
a \emph{generalized unitary quantum walk symmetry} on $\Gamma$ if there is 
a sequence of local unitary operators $(U_n)_{n=0}^{+\infty}$ and a shifted 
automorphism $g\phi$ such that for each quantum coin $\C 
= (C_n)_{n=0}^{+\infty}$ and for each initial state $\ket{\psi_0}$,
\beq
\label{eq-symm2}
\forall n \in \mathbb{N}_0: \quad W_\tC(n) \ket{\tilde\psi_0} = P_{g\phi} U_n 
W_\C(n) \ket{\psi_0},
\eeq
where $\tC$ and $\ket{\tilde\psi_0}$ have the same meaning as in \dfnref{symm} 
and $P_{g\phi}$ denotes the total permutation operator associated with the 
shifted $S$-preserving automorphism $g\phi$.
\end{dfn}

The (unshifted) automorphisms to be considered have to preserve the generating 
set $S$ in order to preserve the edges of the Cayley graph $\Gamma(G,S)$.  We 
note, however, that the automorphism group of $\Gamma(G,S)$ may be more 
general.%
\footnote{Consider, for example, a free group over three generators, $a$, $b$, 
and $c$. The elements are uniquely described by words in the alphabet 
$\mathcal{A} = \lbrace a,b,c,a^{-1},b^{-1},c^{-1}\rbrace$. Define a map 
$\mathcal{A}^\ast \to \mathcal{A}^\ast$ which substitutes $b$ for $c$ and 
\emph{vice versa} for words beginning with an $a$ and leaves all other words 
intact.  Such a map induces a graph automorphism of the Cayley graph but is not 
a group automorphism itself.}
As shown by the following Lemma, the shifted $S$-preserving automorphisms form 
a subgroup of the automorphism group of $\Gamma$.
\begin{lma}
\label{lma-shifted}
The set of all automorphisms on $G$ which preserve $S$ forms a subgroup 
$\AutGS$ of $\mathit{Aut}(G)$.  The set of all shifted $S$-preserving 
automorphisms on $G$ with the operation of map composition forms a group 
isomorphic to $G \rtimes \AutGS$.
\end{lma}

\begin{proof}
For the first part, it suffices to show that for any pair $\phi_1, \phi_2$ of 
automorphisms on $G$ preserving $S$, $\phi_1^{-1} \circ \phi_2$ preserves $S$.  
This is simple as both $\phi_1$ and $\phi_2$ act as permutations when 
restricted to $S$.

To show that the shifted $S$-preserving automorphisms constitute a group, we 
have to prove that the four group axioms are satisfied.

\begin{subequations}
\emph{Closure.} Let $\phi_1, \phi_2 \in \AutGS$ and $g_1, g_2 \in G$. The 
composition of $g_1\phi_1$ and $g_2\phi_2$ is a map $G \to G$ prescribed by
\beq
\label{eq-closure}
(g_1\phi_1 \circ g_2\phi_2)(x) = g_1\phi_1(g_2\phi_2(x)) = g_1\phi_1(g_2) 
\cdot (\phi_1\circ\phi_2)(x).
\eeq
Noting that $g_1\phi(g_2) \in G$ and that $\phi_1\circ\phi_2 \in \AutGS$,
the composed map is by definition a shifted $S$-preserving automorphism.

\emph{Associativity.} Associativity is granted by the operation of 
composition.

\emph{Identity.} The identity element is the shifted $S$-preserving 
automorphism $e\Id$, where $e$ is the identity element in $G$.  Indeed, this 
is the identity map on $G$ and thus the neutral element with respect to map 
composition.

\emph{Inverse.} Let $\phi \in \AutGS$, let $g \in G$. Then the inverse element 
of the shifted $S$-preserving automorphism $g\phi$ with respect to composition 
is a map $G \to G$ defined by
\beq
\label{eq-inverse}
(g\phi)^{-1}(x) = \phi^{-1}\left(g^{-1}x\right) = \phi^{-1}\left(g^{-1}\right) 
\cdot \phi^{-1}(x)
\eeq
This is a shifted $S$-preserving automorphism as $\phi^{-1}(g^{-1}) \in G$ and 
$\phi^{-1} \in \AutGS$.
\end{subequations}

Let us denote this group $\mathcal{G}$. In order to show that $\mathcal{G} 
\cong G \rtimes \AutGS$, we first identify $G$ with a subgroup $G'$ of 
$\mathcal{G}$ using the monomorphism
$$\gamma: G \to \mathcal{G}: g \mapsto g\Id$$
and similarly identify $\AutGS$ with a subgroup $A'$ of $\mathcal{G}$ using 
the monomorphism
$$\alpha: \AutGS \to \mathcal{G}: \phi \mapsto e\phi.$$

It follows directly from the definition that $\mathcal{G} = G'A'$ and that $G' 
\cap A' = \lbrace e \rbrace$. In order to show that the product is semidirect, 
we show that $G'$ is a normal subgroup of $\mathcal{G}$. Let $h\Id \in G'$, 
let $g\phi$ be an arbitrary element of $\mathcal{G}$. Using \eqref{closure} 
and \eqref{inverse}, we simplify the composition
$$\begin{array}{l}
g\phi \circ h\Id \circ (g\phi)^{-1} = g\phi \circ h\Id \circ \phi^{-1}(g^{-1}) 
\phi^{-1} =\\
\qquad = g\phi \circ (h\phi^{-1}(g^{-1}) \phi^{-1} = \\
\qquad = g\phi\left(h\phi^{-1}(g^{-1})\right) (\phi\circ\phi^{-1}) = \\
\qquad = \left(g\phi(h)g^{-1}\right) \Id \in G'.
\end{array}$$
This proves that $\mathcal{G} = G' \rtimes A' \cong G \rtimes \AutGS$.
\end{proof}

\begin{lma}
In the notation of \dfnref{symm2}, the condition of \eqref{symm2} is satisfied 
for each $\C$ and each $\ket{\psi_0}$ if and only if
\beq
\label{eq-symm2-diff}
\begin{array}{c}
\ket{\tilde\psi_0} = P_{g\phi} U_0 \ket{\psi_0}, \\[2pt]
\forall n \in \mathbb{N}_0: \quad T \tilde C_n = P_{g\phi} U_{n+1} TC_n U_n\d 
P_{g\phi}\d.
\end{array}
\eeq
Here, $\tC$ and $\ket{\tilde\psi_0}$ denote the image of $\C$ and 
$\ket{\psi_0}$ under $\T$.
\end{lma}

\begin{proof}
The proof is done in a straightforward analogy to the proof of \lmaref{diff}.
\end{proof}

\begin{lma}
\label{lma-perm}
Let $\phi \in \AutGS$, let $g \in G$. Then the total permutation operator 
$P_{g\phi}$ commutes with the step operator $T$.
Furthermore, let $U$ be a local unitary operation.  Then $P_{g\phi}\d 
U P_{g\phi}$ is a local unitary operation. If $U$ is of the form $\Id\otimes 
U'$, then $P_{g\phi}\d U P_{g\phi}$ is of the form $\Id\otimes 
(P_{g\phi}^{(C)\dagger} U' P_{g\phi}^{(C)})$.
\end{lma}

\begin{proof}
To show the commutation of $T$ and $P_{g\phi}$, we compare the action of both 
$TP_{g\phi}$ and $P_{g\phi}T$ on the same basis state $\ket{x,c}.$
$$\begin{array}{l}
TP_{g\phi} \ket{x,c} = T \ket{g\phi(x),\phi(c)} = \ket{g\phi(x)\phi(c), 
\phi(c)} \\
P_{g\phi}T \ket{x,c} = P_{g\phi} \ket{xc,c} = \ket{g\phi(xc), \phi(c)}
\end{array}$$
The equality $\phi(x)\phi(c) = \phi(xc)$ follows from the fact that $\phi$ is 
a group automorphism.

In order to prove the second part of the Lemma, we first note that all the 
operators $P_{g\phi}^{(S)}$, $P_{g\phi}^{(C)}$, and $P_{g\phi}$ are unitary.  
This can be shown promptly from the fact that the operators act as 
permutations in the corresponding geometrical basis systems. Thus for any 
unitary operator $U$, $P_{g\phi}\d U P_{g\phi}$ is also unitary.

If $U$ is local, we can show using \eqref{perm-total}
$$\begin{array}{l}
\tilde U := P_{g\phi}\d \left( \sum_{x \in G} \ket{x}\bra{x} \otimes U_x 
\right) P_{g\phi} = \\
\qquad = \sum_{x \in G} \left( \left(P_{g\phi}^{(S)\dagger} \ket{x}\bra{x} 
P_{g\phi}^{(S)} \right) \otimes \left(P_{g\phi}^{(C)\dagger} U_x 
P_{g\phi}^{(C)}\right) \right)
\end{array}$$
If we change the summation variable from $x$ to $y = \phi^{-1}(g^{-1}x)$, such 
that $g\phi(y) = x$, we obtain
$$\tilde U = \sum_{y \in G} \left(\!P_{g\phi}^{(S)\dagger} 
\ket{g\phi(y)}\bra{g\phi(y)} P_{g\phi}^{(S)}\!\right) \otimes 
\left(\!P_{g\phi}^{(C)\dagger} U_{g\phi(y)} P_{g\phi}^{(C)}\!\right).$$
We used the fact that the composition of an automorphism and left 
multiplication is a bijection on $G$.

Using the unitarity of $P_{g\phi}^{(S)}$, from which it follows that
$$P_{g\phi}^{(S)\dagger} \ket{g\phi(y)} = \left( P_{g\phi}^{(S)} \right)^{-1} 
\ket{g\phi(y)} = \ket{y}$$
and
$$\bra{g\phi(y)}P_{g\phi}^{(S)} = \left(P_{g\phi}^{(S)\dagger} 
\ket{g\phi(y)}\right)\d = \bra{y},$$
we can simplify $\tilde U$ to the form
$$\tilde U = \sum_{y \in G} \ket{y}\bra{y} \otimes 
\left(P_{g\phi}^{(C)\dagger} U_{g\phi(y)} P_{g\phi}^{(C)} \right),$$
which proves that $\tilde U$ is a local operator.

Similarly, let $U = \Id \otimes U'$. Then
$$\begin{array}{l}
P_{g\phi}\d (\Id \otimes U') P_{g\phi} = \\
\qquad = \left(P_{g\phi}^{(S)\dagger} \Id\, P_{g\phi}^{(S)}\right) \otimes 
\left(P_{g\phi}^{(C)\dagger} U' = P_{g\phi}^{(C)}\right) = \\
\qquad = \Id \otimes \left(P_{g\phi}^{(C)\dagger} U' P_{g\phi}^{(C)}\right).
\end{array}$$

\end{proof}

\enlargethispage{2cm}
As shown by the following Theorem, the search for generalized unitary quantum 
walk symmetries can be reduced to the problem already solved in 
\secref{symm1}.
\begin{thm}
\label{thm-symm2}
Let $\T$ be an endomorphism on the Cartesian product of the set of quantum 
coins and initial states of a quantum walk on $\Gamma$, let $T(\C, 
\ket{\psi_0}) = (\tC, \ket{\tilde\psi_0}), \tC = (\tilde 
C_n)_{n=0}^{+\infty}$.  Then $\T$ is a generalized unitary quantum walk 
symmetry if and only if there is a ordinary unitary quantum walk symmetry 
$\T': (\C, \ket{\psi_0}) \mapsto (\ttC, \ket{\psi_0'}), \ttC 
= (C_n')_{n=0}^{+\infty}$, and a shifted $S$-preserving automorphism $g\phi$ 
such that
\beq
\label{eq-ttcoin}
\begin{array}{c}
\ket{\tilde\psi_0} = P_{g\phi}\ket{\psi_0'} \\[2pt]
\forall n \in \mathbb{N}_0: \quad \tilde C_n = P_{g\phi} C'_n P_{g\phi}\d.
\end{array}
\eeq
\end{thm}

\thmref{symm2} solves in general the problem of symmetries without any 
assumptions about the coin. The restricted problems with position- and/or 
time-independent coins can also be addressed. As a direct consequence of 
\lmaref{perm}, the restriction is transferred from the quantum coin $\C$ to 
the quantum coin $\ttC$ of the original problem, where we can use 
\thmref{indep} or \corref{indep} to find all solutions.

It also trivially follows that the symmetry group is in all cases simply 
augmented by the group of shifted $S$-preserving automorphisms.

\begin{proof}
Let us define $\ket{\psi_0'}$ and $C'_n$ such that \eqref{ttcoin} is held.  
Then, according to \eqref{symm2-diff}, these objects must satisfy
\begin{subequations}
\beq
\label{eq-symm2-pom1}
\ket{\psi_0'} = U_0\ket{\psi_0}
\eeq
and
\beq
\label{eq-symm2-pom2}
TP_{g\phi}C'_nP_{g\phi}\d = P_{g\phi}U_{n+1}TC_nU_n\d P_{g\phi}\d.
\eeq
Using the commutativity of $T$ and $P_{g\phi}$, \eqref{symm2-pom2} becomes
\beq
\label{eq-symm2-pom3}
TC'_n = U_{n+1}TC_nU_n\d.
\eeq
\end{subequations}
However, \eqref{symm2-pom1} and \eqref{symm2-pom3} are exactly the conditions 
of \lmaref{diff}. Therefore $\T$ is a generalized unitary quantum walk 
symmetry if and only if the map $(\C, \ket{\psi_0}) \mapsto (\C', 
\ket{\psi_0'})$ is an ordinary unitary quantum walk symmetry.
\end{proof}

\begin{exm}
We show an application of the generalized quantum walk symmetries again on 
a quantum walk on a line with a homogeneous coin. Given a coin $\C 
= (Id\otimes C)_{n=0}^{+\infty}$ and an initial state $\ket{\psi_0}$, we can 
use \thmref{symm2} to find a new homogeneous quantum coin $\tC = (Id \otimes 
\tilde C)_{n=0}^{+\infty}$ and an initial state $\ket{\tilde\psi_0}$ such that 
the evolution of the new quantum walk is a mirror image of the original one.

Taking the $S$-preserving automorphism $P: x \mapsto -x = 0+(-1)x$, we 
construct the tuple of permutation operators $P_{P}^{(i)}$ easily. We note 
that the matrix of the coin permutation operator is the Pauli $X$-matrix, or 
the quantum \textsc{not} gate.

In the simplest case, we can choose to only perform the permutation, choosing 
the identity transform as $\T'$ in \thmref{symm2}. Doing so, not only the 
measurement probabilities but also the amplitudes are preserved, they only 
undergo the permutation in both position and coin geometrical bases. In this 
case, the transformed coin is described by the matrix
$$\tilde C = XCX\d = XCX$$
and the transformed initial state satisfies
$$\braket{x}{\tilde\psi_0} = X\braket{-x}{\psi_0}$$
for all $x \in \mathbb{Z}$.
\end{exm}

If we use the general form of the coin as described by \eqref{gen-uni}, after 
the transformation we obtain
$$\tilde C =
\omega\begin{matrix}
\mu^\ast & 0 \\ 0 & \mu
\end{matrix} \begin{matrix}
\cos\psi & -\sin\psi \\
\sin\psi & \phantom{-}\cos\psi
\end{matrix} \begin{matrix}
\nu^\ast & 0 \\ 0 & \nu
\end{matrix}.
$$
We note that it is now possible, if desired, to transform the coin back to its 
original state, using the results of \secref{symm1} only. This way, the 
probability distribution stays unchanged, i.e. mirrored with respect to the 
original quantum walk, thus we obtain a new initial state $\ket{\psi_1}$ for 
the original coin $C$ for which the time evolution has flipped sides.

We can do so by the following transform:
$$\begin{array}{l}
C = \omega
\begin{matrix}
\mu^\ast & 0 \\ 0 & \mu
\end{matrix} \begin{matrix}
\mu^2\nu^2 & \hfil 0 \\ \hfil 0 & (\mu^2\nu^2)^\ast
\end{matrix} \begin{matrix}
-i\nu^{\ast2} & \hfil 0 \\ \hfil 0 & i\nu^2
\end{matrix} \times \\ \quad \begin{matrix}
\cos\psi & -\sin\psi \\
\sin\psi & \phantom{-}\cos\psi
\end{matrix} \begin{matrix}
i\nu^2 & \hfil 0 \\ \hfil 0 & -i\nu^{\ast2}
\end{matrix} \begin{matrix}
\nu^\ast & 0 \\ 0 & \nu
\end{matrix}.
\end{array}
$$
This corresponds to choosing $\delta_- = \delta_+^\ast = i\nu^2$, $e^{i\phi} 
= \mu^2 \nu^2$, and $\epsilon = 1$ in the notation of the Example in 
\secref{symm1}. The choice of $\eta_{\mathrm{even}}$ and $\eta_{\mathrm{odd}}$ 
is free, so we can let them be $1$. The transformed initial state is then 
given by
$$\ket{\psi_1} = \sum_{x \in \mathbb{Z}}
\begin{matrix}
-i\nu^{\ast2} & \hfil 0 \\ \hfil 0 & i\nu^2
\end{matrix}
X \ket{{-x}}\braket{x}{\psi_0}.$$

If the initial state $\ket{\psi_0}$ is localized at $x=0$, the transition to 
$\ket{\psi_1}$ is simply a linear transformation of the initial chirality, 
described in the geometrical basis by the matrix
$$Q = \begin{matrix}
-i\nu^{\ast2} & \hfil 0 \\ \hfil 0 & i\nu^2
\end{matrix} X =
\begin{matrix}
\hfil 0 & -i\nu^{\ast2} \\ i\nu^2 & \hfil 0
\end{matrix}$$

Having this result enables us to find initial states which produce a symmetric 
probability distribution at each iteration of the quantum walk. These are 
simply the eigenstates of the matrix $Q$, tensor multiplied by $\ket{0}$ in 
the position register. The eigenvalues of $Q$ are $\pm 1$ and the 
corresponding normalized eigenvectors are
$$\ket{\chi_0}_\pm = \frac{1}{\sqrt{2}}
\begin{matrix} \hfil \nu^\ast \\ \pm i\nu \end{matrix}$$
in the coin space basis. Except for the degenerate cases of $\psi = k\pi$, $k 
\in \mathbb{Z}$, the parameter $\nu$ is defined uniquely up to a sign and 
therefore there are exactly two localized initial states producing a symmetric 
probability distribution and these are orthonormal.

\section{Conclusions}
\label{sec-concl}

We used analytic and algebraic methods to study the symmetries of discrete 
time quantum walks on Cayley graphs, where the quantum coin was allowed to 
transform along with the initial state. We constructed a general way of 
obtaining transformations which preserve the measurement probabilities, and 
our results grant that we obtained the complete set of such transformations in 
a uniform manner. We described the symmetry group of the quantum walk time 
evolution operator using the results of the analysis.

Some of the symmetries found this way correspond to trivial properties of any 
discrete time quantum system, but most of the symmetries are specific to 
quantum walks. Once the symmetry group is found, any continuous symmetry can 
be used to reduce the problem. We have demonstrated this fact on the quantum 
walk on a line with a constant coin, where the result was that two out of 
three physical parameters of the quantum coin could be dropped without loss of 
generality. Quantum walks on more complicated graphs allow even more 
significant reduction.

An open question is how the results change if we drop the condition that the 
Cayley graph is nonseparating. An example where this condition is not held is 
a quantum walk on any group which contains the free group of order $2$ or 
higher. Counterexamples to the forms provided by \thmref{indep} can be found 
for such graphs, indicating that a more general treatment is necessary to 
cover all Cayley graphs.

However, the most important open question, which could be addressed in 
a subsequent work, is how the results change if the definition of a quantum walk 
is generalized such that the dimension of the coin space is different from the 
out-degree of the Cayley graph.

\begin{acknowledgments}
This work was supported by the Grant Agency of the Czech Technical University in 
Prague, grant No. SGS10/294/OHK4/3T/14.
\end{acknowledgments}


\begin{thebibliography}{99}

\bibitem{chandra}
C. M. Chandrashekar, R. Srikanth, and S. Banerjee, ``Symmetries and noise in 
quantum walk'', Phys. Rev. A {\bf 76}, 022316 (2007)

\bibitem{sym-hit}
H. Krovi, ``Symmetry in quantum walks'', Ph.D. thesis, University of Southern 
California (2007)

\bibitem{skw}
N. Shenvi, J. Kempe, and K. B. Whaley, ``Quantum random-walk search 
algorithm'', Phys. Rev. A {\bf 67}, 052307 (2003)

\bibitem{exp}
A. Schreiber \emph{et al.}, ``Photons Walking the Line: A Quantum Walk with 
Adjustable Coin Operations'', Phys. Rev. Lett. {\bf 104}, 050502 (2010)

\bibitem{lattice2}
A. Ambainis, J. Kempe, and A. Rivosh, ``Coins make quantum walks faster'', in 
\emph{Proceedings of the Sixteenth Annual ACM-SIAM Symposium on Discrete 
Algorithms} (2005), pp. 1099--1108

\bibitem{first}
D. Aharonov, A. Ambainis, J. Kempe, and U. Vazirani,
``Quantum walks on graphs'', in \emph{Proceedings of the Thirty-third ACM 
Symposium on Theory of Computing} (2001), pp. 50--59

\bibitem{no-go}
D. Meyer, ``On the absence of homogeneous scalar unitary cellular automata'',
Phys. Lett. A {\bf 223}, pp. 337--340 (1996)

\bibitem{cayley-scalar}
O. L. Acevedo, J. Roland, and N. J. Cerf, ``Exploring scalar quantum walks on 
Cayley graphs'', Quant. Inform. Comp. {\bf 8} (2008), pp. 68--81

\bibitem{cayley}
O. L. Acevedo and T. Gobron, ``Quantum walks on Cayley graphs'', J. Phys. A: 
Math. Gen. 39 (2006) 585-599

\end{thebibliography}
\end{document}